\documentclass[10pt,journal,twocolumn,twoside]{IEEEtran}
\usepackage{cite}
\usepackage{caption}
\usepackage{array}
\usepackage{amsmath,amssymb,amsfonts}
\usepackage{setspace}
\usepackage{algorithmic}
\usepackage{graphicx}
\usepackage{tabularx}
\usepackage{booktabs}
\usepackage{setspace}
\usepackage{multirow}
\usepackage{fancyhdr}
\usepackage{float}
\usepackage{textcomp}
\usepackage{xcolor}
\usepackage{amsthm}
\usepackage{graphics}
\usepackage{subfigure}
\usepackage{subfig}
\usepackage{algorithm}
\usepackage{algorithmic}
\usepackage{etoolbox}
\usepackage{cleveref}

\theoremstyle{plain}
\newtheorem{theorem}{Theorem}
\theoremstyle{plain}
\newtheorem{proposition}{Proposition}
\theoremstyle{plain}

\theoremstyle{plain}
\newtheorem{lemma}{Lemma}
\theoremstyle{plain}

\ifCLASSINFOpdf
\else
\fi
\begin{document}
\title{Superdirectivity-enhanced Wireless Communications:
A Multi-user  Perspective}
\author{Liangcheng Han,
        and Haifan Yin,~\IEEEmembership{Member,~IEEE}
      
\thanks{Liangcheng Han and Haifan Yin are with the School of Electronic Information and Communications, Huazhong University of Science and Technology,
Wuhan 430074, China (e-mail: hanlc@hust.edu.cn; yin@hust.edu.cn;). The corresponding author is Haifan Yin.}
\thanks{This work was supported by the National Natural Science Foundation of China under Grant 62071191.}
}

\maketitle
\begin{abstract}
Superdirective array may achieve an array gain proportional to the square of the number of antennas $M^2$. In the early studies of superdirectivity, little research has been done from wireless communication point of view. To leverage superdirectivity for enhancing the spectral efficiency, this paper investigates multi-user communication systems with superdirective arrays.  We first propose a field-coupling-aware (FCA) multi-user channel estimation method, which takes into account the antenna coupling effects. Aiming to maximize the power gain of the target user, we propose multi-user multipath superdirective precoding (SP) as an extension of our prior work on coupling-based superdirective beamforming. Furthermore, to reduce the inter-user interference, we propose interference-nulling superdirective precoding (INSP) as the optimal solution to maximize user power gains while eliminating interference. Then, by taking the ohmic loss into consideration, we further propose  a regularized interference-nulling superdirective precoding (RINSP) method. Finally, we discuss the well-known narrow directivity bandwidth issue, and find that it is not a fundamental problem of superdirective arrays in multi-carrier communication systems. Simulation results show our proposed methods outperform the state-of-the-art methods significantly. Interestingly, in the multi-user scenario, an 18-antenna superdirective array can achieve up to a 9-fold increase of spectral efficiency compared to traditional multiple-input multiple-output (MIMO), while simultaneously reducing the array aperture by half.
\end{abstract}

\begin{IEEEkeywords}
superdirectivity, compact arrays, multi-user precoding, channel estimation, wideband beamforming.
\end{IEEEkeywords}
\IEEEpeerreviewmaketitle
\section{Introduction}\label{I} 
\IEEEPARstart{M}{assive} multiple-input multiple-output (MIMO) has emerged as one of the key technologies for the 5th generation (5G) communication systems \cite{marzetta2010noncooperative}. By deploying a large number of antennas at the base station, it may lead to a significant improvement in spectral efficiency. Using basic precoding techniques like maximum ratio transmission (MRT) or zero-forcing (ZF), massive MIMO empowers spatial division multiple access (SDMA) to serve multiple users simultaneously on the same frequency\cite{rusek2012scaling}. However, to reduce the mutual coupling in massive MIMO systems, the antenna spacing is no less than half a wavelength. This constraint further limits the increase in the number of antennas and the improvement of spectral efficiency. In recent years, with the pursuit of  higher spectral efficiency, compact arrays have attracted widespread attention. 
The concept of Holographic MIMO\cite{bjornson2019massive,pizzo2022fourier,pizzo2020spatially,huang2020holographic}  has appeared in the vision of wireless communication researchers. In Holographic MIMO, the locations of antennas is no longer discrete; all antennas will be continuously distributed on the antenna panel. Such surfaces consisting of densely populated electromagnetic excitation elements are able to record and manipulate incident fields with maximum flexibility and precision, while reducing cost and power consumption, thereby forming arbitrary intended electromagnetic waves with high energy efficiency\cite{gong2022holographic}.
In fact, the strong mutual coupling introduced by compact arrays may endow superdirectivity to the antenna array, achieving an array gain of $M^2$ instead of $M$ with the traditional MIMO\cite{marzetta2019super}.

Directivity reflects the ability of an array to focus electromagnetic radiation in a particular direction. Superdirectivity refers to the ability of an antenna array to produce a radiation pattern with a substantially narrower beamwidth and higher gain compared to a regular antenna array \cite{Collin_Zucker_1969}. In 1946, Uzkov has proved that the directivity of an isotropic linear array of $M$ antennas can reach $M^2$ when the spacing between antennas approaches zero\cite{uzkov1946approach}. 
Despite the early introduction of superdirectivity, researchers have mainly focused on single-user scenarios from the perspective of the antenna array theories \cite{altshuler2005monopole,o2006electrically,hansen2006electrically,clemente2015design,jaafar2018characteristic}, which are different from communication theories.  To the best of our knowledge, the research on applying superdirective antenna arrays to multi-user communication systems to improve the spectral efficiency of wireless communication systems has not been demonstrated in the open literature. This paper will conduct a detailed study on this topic.

In our previous work \cite{han2023superdirective}, we investigated how to implement superdirective antenna arrays in practice and found that the field coupling introduced in compact arrays leads to strong distortions in the antenna radiation patterns. Therefore, the channel model in compact arrays needs to be corrected based on the distortion of the array response caused by field coupling. To this end, we propose a coupling channel model for compact arrays and present a field-coupling-aware (FCA) multi-user channel estimation method. Based on the estimated channel, we propose a series of superdirective multi-user precoding schemes. First, we propose the multi-user  superdirective  precoding (SP) scheme, which extends our previous approach \cite{DBLP:conf/icc/HanYM22} of coulpling-based superdirective beamforming  to multi-user multipath scenarios. The SP scheme aims to maximize the power gain for each user. However, the performance of this scheme may degrade when user interference is strong. Thus, we formulate a problem to maximize target user power gain while nulling interference in superdirective arrays and derive the optimal solution, which is called interference-nulling superdirective precoding (INSP) scheme. Considering the ohmic loss in superdirective arrays, we propose regularized interference-nulling superdirective precoding (RINSP), and we provide the power gain expression under a given antenna loss. Additionally, to address the well-known narrow directivity bandwidth issue in superdirective arrays, we propose a superdirective wideband multi-user beamforming method.

Specifically, the contributions of this paper are as follows:
\begin{itemize}
    \item To the best of our knowledge, the paper is the first to study superdirective antenna arrays in multi-user communication systems, and we demonstrate a significant improvement in spectral efficiency compared to traditional  MIMO.
    \item We propose an FCA channel estimation method that considers the channel distortion caused by antenna field coupling. Simulation results show that, compared to the traditional channel estimation method that ignores field coupling, the FCA estimation scheme can significantly reduce channel estimation errors.
    \item For multi-user scenarios, we initially propose SP scheme to maximize user power gain. To further enhance performance and specifically eliminate user interference, we then introduce INSP scheme for superdirective arrays. Compared to MRT and ZF, the proposed precoding methods achieve a significant improvement in spectral efficiency.
    \item To alleviate the issue of ohmic loss in superdirectivity-enhanced wireless communications, we propose a novel scheme, namely RINSP. We also derive the corresponding expression for power gain. Simulations show that the RINSP brings obvious spectral efficiency gains compared to existing method  for superdirective antenna arrays with ohmic loss.
    \item We analyze the cause of the narrow directivity bandwidth in superdirective arrays and propose a wideband multi-user superdirective precoding method. We show that the narrow bandwidth issue is not as fundamental as it was thought to be in modern multi-carrier wideband wireless communication systems. 
\end{itemize}

The organization of this paper is as follows:
In Sec. \ref{II}, we first briefly review superdirective beamforming and then propose  FCA channel estimation method. In Sec. \ref{III}, we present three precoding methods for superdirective multi-user transmission. In Sec. \ref{IV} , we analyze the cause of the narrow directivity bandwidth in superdirective arrays and propose a wideband multi-user superdirective precoding method. Finally, simulation results are presented in Sec. \ref{V}.

Notation: Lower-case and upper-case boldface letters represent vectors and matrices, respectively; $(\cdot)^*,(\cdot)^T,(\cdot)^H,\|\cdot\|_F$, and $\|\cdot\|_k$ denote the conjugate, transpose, conjugate transpose, Frobenius norm, and $k$-norm of a matrix, respectively; $|\cdot|$ denotes the absolute operator; $\mathbf{I}_N$ represents the identity matrix of size $N \times N$; $\mathbf{X} \otimes \mathbf{Y}$ is the Kronecker product of $\mathbf{X}$ and $\mathbf{Y}$; $\operatorname{vec}(\mathbf{X})$ is the vectorization of the matrix $\mathbf{X}$; $\operatorname{diag}\left\{\mathbf{a}_1, \ldots, \mathbf{a}_{\mathbf{N}}\right\}$ denotes a diagonal matrix or a block diagonal matrix with $\mathbf{a}_1, \ldots, \mathbf{a}_{\mathbf{N}}$ at the main diagonal; $\triangleq$ is used for definition.

\section{System models}\label{II}
In this section, we review the superdirective beamforming and introduce the channel model for compact arrays. In particular, we propose a channel estimation method that considers the field coupling effect.
\subsection{Superdirective beamforming}\label{Superdirective_beamforming}
For ease of exposition, we consider a compact uniform linear array (ULA) formed by isotropic antennas at the base station, with an inter-antenna spacing of $d$ (i.e., $d$ is smaller than half a wavelength) and a total number of antennas $M$. As the antenna spacing is small, the field coupling effect between antennas needs to be taken into account\cite{han2023superdirective}. In this case, the far-field radiation pattern $l(\theta,\phi)$ of the array can be expressed as
\begin{align}
l(\theta, \phi)=\sum_{m=1}^M \sum_{n=1}^N c_{nm} a_m e^{j k \hat{\mathbf{r}} \cdot \mathbf{r}_n},
\end{align}
where $c_{nm}$ is the coupling coefficient between the $n$-th and the $m$-th antennas, $a_m$ is the beamforming coefficient of the $m$-th antenna, $k=\frac{2\pi}{\lambda}$ is the wavenumber, $\lambda$ is the wavelength, $\mathbf{r}_n$ is the spatial coordinate of the $n$-th antenna, and $\hat{\mathbf{r}}$ is the unit vector in the spherical coordinate system, which can be expressed as
\begin{align}
\hat{\mathbf{r}} \triangleq\begin{bmatrix}
\sin \theta \cos \phi\
\sin \theta \sin \phi \
\cos \theta
\end{bmatrix}^T,
\end{align}
where $\theta$ and $\phi$ are the elevation and azimuth angles, respectively. According to the definition in \cite{stutzman2012antenna}, the directivity factor of the array can be expressed as
\begin{align}\label{D}
D(\theta,\phi)=\frac{(\mathbf{C a})^T \mathbf{e}(\theta,\phi) \mathbf{e}^H(\theta,\phi)(\mathbf{C a})^*}{(\mathbf{C a})^T \mathbf{Z}(\mathbf{C a})^*},
\end{align}
where $\mathbf{C}\in \mathbb{C}^{M\times M}$ is the field coupling matrix:
\begin{equation}\label{Cmatrix}
\mathbf{C}\triangleq\left[\begin{array}{ccc}
c_{11} & \ldots & c_{1 M} \\
\vdots & \ddots & \vdots \\
c_{M 1} & \ldots & c_{M M}
\end{array}\right].
\end{equation}
The specific calculation approach for the field coupling matrix $C$ can be found in \cite{han2023superdirective}. $\mathbf{a}$ and $\mathbf{e}(\theta,\phi)$ are the beamforming and steering vectors, respectively:
\begin{align}
\mathbf{a} &= \left[a_1, a_2, \cdots, a_M\right]^T, \\
\mathbf{e}(\theta,\phi) &= \left[e^{j k \hat{\mathbf{r}} \cdot \mathbf{r}_1} , e^{j k \hat{\mathbf{r}} \cdot \mathbf{r}_2} , \cdots, e^{j k \hat{\mathbf{r}} \cdot \mathbf{r}_M} \right]^T.
\end{align}
$\mathbf{Z}\in \mathbb{R}^{M\times M}$ is the impedance coupling matrix, and the $(m,n)$-th element of $\mathbf{Z}$, is given by
\begin{align}\label{zmn}
z_{m n} = \frac{1}{4 \pi} \int_0^{2 \pi} \int_0^\pi e^{j k \hat{\mathbf{r}} \cdot \mathbf{r}_m} e^{-j k \hat{\mathbf{r}} \cdot \mathbf{r}_n} \sin \theta d \theta d \phi
\end{align}
For isotropic antennas, a closed-form expression for $z_{mn}$ in \eqref{zmn} is available and is given by
\begin{equation}
z_{m n}=\left\{\begin{array}{cc}
\frac{\sin [k d(n-m)]}{k d(n-m)}, & n \neq m \\
1, & n=m
\end{array}\right.,
\end{equation}
For the convenience of analysis, we rewrite \eqref{D} as
\begin{align}\label{D2}
D(\theta,\phi)=\frac{\bar{\mathbf{a}}^T \mathbf{e}(\theta,\phi) \mathbf{e}^H(\theta,\phi)\bar{\mathbf{a}}^*}{\bar{\mathbf{a}}^T\mathbf{Z}\bar{\mathbf{a}}^*},
\end{align}
where $\bar{\mathbf{a}}=\mathbf{C}\mathbf{a}$, and the conversion from $\bar{\mathbf{a}}$ to $\mathbf{a}$ can be easily done by taking the inverse of $\mathbf{C}$. The physical meanings of the numerator and denominator in \eqref{D2} are the radiation power in the $(\theta,\phi)$ direction and the average radiation power over the entire space, respectively. Maximizing \eqref{D2} with respect to $\bar{\mathbf{a}}$ yields\cite{han2023superdirective}
\begin{align}\label{bara}
\bar{\mathbf{a}}=\gamma \mathbf{Z}^{-1} \mathbf{e}^*(\theta,\phi),
\end{align}
where $\gamma$ is the power constraint coefficient. The maximum directivity factor is expressed by
\begin{align}
D_{\max}(\theta,\phi)=\mathbf{e}^H(\theta,\phi) \mathbf{Z}^{-1} \mathbf{e}(\theta,\phi).
\end{align}
\subsection{Field coupling aware channel estimation (FCA)}
The superdirective antenna array requires a compact arrangement of antennas. In this case, there is significant field coupling between the antennas, which distorts the radiation pattern of the antennas\cite{han2023superdirective}. The distorted  radiation pattern will cause variations in the amplitude and phase of the array response, which is a critical factor in beamforming and affects the performance of antenna arrays. Obviously, neglecting the field coupling will result in severe channel estimation errors. 
Our previous work \cite{han2023superdirective} has shown that this field coupling between the antennas can be accurately described by the field coupling matrix. In the following modeling, we will introduce the field coupling matrix to make the model more closely match the real physical channel.

Without loss of generality, assume that there are $U$ users in the system, each equipped with a single antenna. The base station operates in time-division duplexing (TDD) mode, which means that the uplink and downlink channels are reciprocal. 
In order to obtain channel state information (CSI), each user is assigned a pilot sequence of length $\tau$. The pilot sequence $\mathbf{s}_u$ for user $u$ is denoted by:
\begin{align}
    \mathbf{s}_u=\left[\begin{array}{llll}
s_{u 1} & s_{u 2} & \cdots & s_{u \tau}
\end{array}\right]^T .
\end{align}
Suppose that the pilot power of each user is equal and satisfies $\left|s_{u1}\right|^2+\cdots+\left|s_{u\tau}\right|^2=\tau$ for $u=1,2,\ldots,U$. 

For a coherent time-frequency resource block, the channel vector between the $u$-th user and the base station can be considered invariant and is denoted by $\mathbf{h}_u$, and $\mathbf{R}_u\triangleq\mathbb{E}\{\mathbf{h}_u\mathbf{h}_u^H\}$ represents the covariance matrix of the channel. 
The multi-path channel of user $u$ can be written as \cite{tse2005fundamentals}
\begin{align}\label{multipath_h}
    \mathbf{h}_u=\frac{1}{{P}} \sum_{p=1}^P \mathbf{e}\left(\theta_{u p},\phi_{up}\right) \alpha_{u p}
\end{align}
where $\alpha_{up}\sim \mathcal{C N}\left(0, \delta_u^2\right)$ represents the complex channel gain of the path $p$ from the base station to user $u$,  where $\delta_u$ is the channel average attenuation of the $u$-th channel, and $P$ is the arbitrary number of paths.

During the CSI acquisition phase, if the field coupling between antennas is not taken into account, the received signal at the base station can be expressed as
\begin{align}
   \mathbf{Y}=\sum_{u=1}^U \mathbf{h}_u \mathbf{s}_u^T+\mathbf{N}, 
\end{align}
where  $\mathbf{N}$ is a $M\times \tau$ additive white Gaussian noise (AWGN) matrix with a zero-mean and element-wise variance of $\epsilon_n^2$. 

However, due to the presence of field coupling between the antennas, the signal received by the array is distorted, thereby the actual channel model needs to be modified accordingly. 
Thus, the actual received signal is given by:
\begin{align}\label{receive_coupling}
   \mathbf{Y}_{\mathrm{c}}=\sum_{u=1}^U \mathbf{C}^T\mathbf{h}_u \mathbf{s}_u^T+\mathbf{N}. 
\end{align}

By vectorizing the received signal and noise, our model \eqref{receive_coupling} can also be written as:
\begin{align}\label{vector_Y}
    \mathbf{y}_c=\tilde{\mathbf{S}} \mathbf{h}+\mathbf{n},
\end{align}
where $\mathbf{y}_c=\operatorname{vec}(\mathbf{Y}_c)$, $\mathbf{n}=\operatorname{vec}(\mathbf{N})$, and $\mathbf{h} \in \mathbb{C}^{UM\times 1}$ is obtained by  concatenating all $U$ channels into a vector. The pilot matrix $\tilde{\mathbf{S}}$ is defined as
\begin{align}
   \tilde{\mathbf{S}} \triangleq \begin{bmatrix}
\mathbf{s}_1 \otimes \mathbf{C}^T& \cdots & \mathbf{s}_U \otimes \mathbf{C}^T
\end{bmatrix}. 
\end{align}
By applying the minimum mean square error (MMSE) estimation to \eqref{vector_Y} \cite{yin2013coordinated}, an estimate of the multi-user channel $\mathbf{h}$ is given by: 
\begin{equation}
   \begin{aligned}\label{EFCA}
    \widehat{\mathbf{h}}&=\mathbf{R} \tilde{\mathbf{S}}^H\left(\tilde{\mathbf{S}} \mathbf{R} \tilde{\mathbf{S}}^H+\epsilon_n^2 \mathbf{I}_{\tau M}\right)^{-1} \mathbf{y}_c,
\end{aligned}
\end{equation}
where $\mathbf{R} \triangleq \operatorname{diag}\left(\mathbf{R}_1, \cdots, \mathbf{R}_L\right)$. Thus, the proposed field coupling aware (FCA) estimator is
\begin{equation}
    \mathbf{E}_{\mathrm{FCA}}\triangleq\mathbf{R} \tilde{\mathbf{S}}^H\left(\tilde{\mathbf{S}} \mathbf{R} \tilde{\mathbf{S}}^H+\epsilon_n^2 \mathbf{I}_{\tau M}\right)^{-1}.
\end{equation}

For the traditional estimator that neglects field coupling, the expression is given by:
\begin{equation}\label{Etrad}
    \mathbf{E}_{\mathrm{Trad}} \triangleq \mathbf{R} {\mathbf{S}}^H\left({\mathbf{S}} \mathbf{R} {\mathbf{S}}^H+\epsilon_n^2 \mathbf{I}_{\tau M}\right)^{-1}.
\end{equation}
where $\mathbf{S}$ is defined as:
\begin{equation}
{\mathbf{S}} \triangleq\left[\begin{array}{lll}
\mathbf{s}_1 \otimes \mathbf{I} & \cdots & \mathbf{s}_U \otimes \mathbf{I}
\end{array}\right].
\end{equation}
By applying the proposed FCA channel estimation method, the CSI of compact arrays can be acquired, which is the basis of superdirective multi-user beamforming.
\section{Superdirective multi-user precoding}\label{III}
In this section, we investigate the multi-user precoding in superdirective antenna arrays.
Our analysis is based on the assumption that all CSI is known, which can be acquired through the method described in the previous section.
\subsection{Multi-user transmission with superdiretive precoding (SP)}
Without loss of generality, suppose the antennas are arranged along the $x$-axis of the coordinate system, and the position of the first antenna is the origin of the coordinate system. The antenna spacing is $d$. Users are all located on the $xy$ plane, i.e., $\theta = 90^\circ$. The steering vector at angle $\phi_{up}$ is
\begin{equation}
\mathbf{e}(\phi_{up}) =\left[\begin{array}{c}
1 \\
e^{-j 2 \pi \frac{d}{\lambda} \cos (\phi_{up})} \\
\vdots \\
e^{-j 2 \pi \frac{(M-1)d }{\lambda} \cos (\phi_{up})}
\end{array}\right],
\end{equation}
Referring to \eqref{D2}, the power gain that user $u$ can obtain with precoding vector $\mathbf{w}_u$ is given by:
\begin{align}\label{Du}
    D_u=\frac{(\mathbf{w}^{\rm SP}_u)^T \mathbf{h}_u\mathbf{h}_u^H{(\mathbf{w}^{\rm SP}_u)^*}}{{(\mathbf{w}^{\rm SP}_u)}^T \mathbf{Z}{(\mathbf{w}^{\rm SP}_u)^*}}.
\end{align}

Unlike the traditional superdirectivity beamforming mentioned in section \ref{Superdirective_beamforming}, which only considers directivity maximization in one direction, the user channel in \eqref{Du} contains multiple paths, i.e., multiple directions. The optimal precoding vector to maximize the user power gain is given by the following lemma.
\begin{lemma}\label{lemma1}
The precoding vector that maximizes the power gain for user $u$ is given by:
\begin{align}\label{wusp}
    \mathbf{w}_u^{\rm SP} = \gamma_u \mathbf{Z}^{-1}\mathbf{h}_u^*,
\end{align}
where $\gamma_u$ is the power control coefficient. $\gamma_u$ is determined by the following equation, for a given radiation power $\rho_u$:
\begin{align}
    \gamma_u=\sqrt{\frac{2\rho_u}{\mathbf{h}_u^H(R_{\rm rad}\mathbf{Z})^{-1}\mathbf{h}_u}},
\end{align}
\end{lemma}
\begin{proof}
\noindent\emph{Proof:} See Appendix A. 
\end{proof}

 The array radiation power $P_{\rm rad}$ with the precoding vector $\mathbf{w}_u^{\rm SP}$ is
\begin{align}\label{Prad}
    P_{\rm rad} = \frac{1}{2}(\mathbf{w}_u^{\rm{SP}})^T(R_{\rm rad}\mathbf{Z})(\mathbf{w}_u^{\rm{SP}})^*,
\end{align}
where $R_{\rm rad}$ is the radiation resistance of the antenna, which is a fixed value that can be calculated using the method of moments (MoM)\cite{Collin_Zucker_1969}.
Therefore, our first proposed superdirectivity multi-user precoding scheme is to calculate the  precoding vector for each user separately based on  \eqref{wusp}, aiming to maximize the power gain of each user.

Note that the first scheme only maximizes the power gain  of the target user, without imposing any restrictions on the signals of interference users. This can improve spectral efficiency in low signal-to-noise ratio (SNR) scenarios where the signals of interference users are weak. Yet when the signals of interference users dominate, the signal-to-interference-plus-noise ratio (SINR) of the target user  will significantly deteriorate, thereby reducing the overall system spectral efficiency. To this end, we then propose superdiretcive multi-user precoding schemes that impose restrictions on the signals of interference users.
\subsection{Multi-user transmission with interference-nulling superdiretive precoding (INSP)}
To maximize the directivity gain of  user $u$ without creating interference to other users, the following mathematical problem needs to be solved:
\begin{equation}
    \begin{array}{ll}\label{P1}
\max\limits_{\mathbf{w}_u} & \dfrac{\mathbf{w}_u^T \mathbf{h}_u \mathbf{h}_u^H \mathbf{w}_u^*}{\mathbf{w}_u^T \mathbf{Z} \mathbf{w}_u^*} \\
\text { s.t. } & \mathbf{h}_i^T \mathbf{w}_u=0, i=1, \cdots, U, i \neq u
\end{array}
\end{equation}

It is noted that the objective function and constraints of \eqref{P1} are convex. Thus \eqref{P1} is a convex optimization problem that can be solved using methods such as the interior point method and gradient descent method. However, these methods require iteration and have high complexity, which are not suitable for fast time-varying channel. Therefore, we propose a low-complexity solution based on basis transformation for problem \eqref{P1}. The specific procedure is as follows. 

For user $u$, the channel covariance matrix $\mathbf{R}^u_{\rm int}$ of all interference users  is written as
\begin{align}\label{Rint}
    \mathbf{R}^u_{\rm int}\triangleq \sum_{i\neq u}\mathbf{h}_i\mathbf{h}_i^H,
\end{align}
and its eigen-value decomposition (EVD) is
\begin{align}\label{decomp}
    \mathbf{R}^u_{\rm int}=\mathbf{L}_u \boldsymbol{\Sigma} \mathbf{L}_u^H,
\end{align}
where $\mathbf{L}_u$ is the eigenvector matrix such that $\mathbf{L}_u^H \mathbf{L}_u=\mathbf{I}_M$.
\begin{algorithm}[tbp]
	\caption{Proposed multi-user interference-nulling superdirective precoding method}
	\label{alg2}
	\begin{algorithmic}[1]
		\REQUIRE ~~\\
		The CSI matrix $[\mathbf{h}_1,\mathbf{h}_2,\cdots,\mathbf{h}_U]$; Power allocation vector $[\rho_1,\rho_2,\cdots,\rho_U]$; The impedance coupling matrix $\mathbf{Z}$;\\
		\ENSURE ~~\\
		Precoding vector \\
		\FOR{$\mathbf{h}_u\in[\mathbf{h}_1,\mathbf{h}_2,\cdots,\mathbf{h}_U]$}		
		\STATE 
		Compute the channel covariance matrix $\mathbf{R}^u_{\rm int}$ as in \eqref{Rint};\\
		\STATE Obtain $\mathbf{L}_u$ by decomposing $\mathbf{R}^u_{\rm int}$ as in \eqref{decomp};\\
		\STATE Obtain $\boldsymbol{\Xi}_u$ and $\boldsymbol{\eta}_u$ respectively as in \eqref{xiu} and \eqref{etau};\\
        \STATE Compute $\boldsymbol{\alpha}_u$ according to \eqref{alphau};
        \STATE Compute $\bar{\mathbf{w}}_u^{\mathrm{I}}$ as in \eqref{barw} and impose power constraint according to \eqref{Prad};
		\ENDFOR 
		\RETURN $[\bar{\mathbf{w}}_u^{\mathrm{I}},\bar{\mathbf{w}}_u^{\mathrm{I}},\cdots,\bar{\mathbf{w}}_u^{\mathrm{I}}]$.
	\end{algorithmic}
\end{algorithm}
\begin{theorem}\label{theorem1}
Denote the dimension of the interference user space by  $N (N\leq U-1)$.
The optimal solution to problem \eqref{P1} is given by
\begin{equation}\label{barw}
       \bar{\mathbf{w}}^{\rm I}_u = \mathbf{L}_u^*\left(\begin{array}{c}
\mathbf{0}_{N}\\
\boldsymbol{\alpha}_u 
\end{array}\right),
\end{equation}
where $\mathbf{0}_{N}$ is a zero vector of size $N \times 1$, and $\boldsymbol{\alpha}_u \in \mathbb{C}^{(M-N) \times 1}$ is obtained by
\begin{align}
    \boldsymbol{\alpha}_u=\boldsymbol{\Xi}_u^{-1} \boldsymbol{\eta}_u^*,
\end{align}
where $\boldsymbol{\Xi}_u \in \mathbb{C}^{(M-N) \times(M-N)}$ and $\boldsymbol{\eta}_u \in \mathbb{C}^{(M-N) \times 1}$ are given respectively in
\begin{equation}\label{xiu}
\mathbf{L}_u^H \mathbf{Z} \mathbf{L}_u=\left(\begin{array}{ll}
\boldsymbol{\Upsilon} & \mathbf{\Psi} \\
\boldsymbol{\Psi} & \boldsymbol{\Xi}_u
\end{array}\right)
\end{equation}
and
\begin{equation}\label{etau}
\mathbf{L}_u^H \mathbf{h}_u=\left(\begin{array}{c}
\boldsymbol{\gamma} \\
\boldsymbol{\eta}_u
\end{array}\right).
\end{equation}
\end{theorem}
\begin{proof}
\noindent\emph{Proof:} See Appendix B.
\end{proof}

Note that although the above method can obtain the optimal solution to problem \eqref{P1}, if the interfering users are close to the target user channel, i.e., the interference channels are highly correlated with the channel of the target user, the  gain obtained when using the optimal beamforming coefficients for precoding may still be poor. It is worth exploring how to use suitable multiple access schemes, e.g., orthogonal frequency division multiple access (OFDMA), SDMA, for each user to maximize system throughput under superdirective arrays. However, this topic is beyond the scope of this paper.

The method is summarized in \textbf{Algorithm} \ref{alg2}.

\subsection{Multi-user transmission with regularized interference-nulling superdiretive precoding (RINSP)}
Deploying a superdirective antenna array in practical communication systems may encounter two of the major challenges: 1) significant ohmic loss and 2) high precision requirements for the precoding vector. This is mainly due to the tendency of the matrix $\mathbf{Z}$ towards singularity when the antenna spacing is small. 
In order to alleviate these challenges and enhance the practicality of superdirective antenna arrays in multi-user systems, we propose a regularized interference-nulling superdirective multi-user precoding.

\begin{proposition}\label{propositionohimc}
    For an antenna array with each antenna having a radiation resistance of $R_{\rm rad}$ and an ohmic loss resistance of $r_{\rm loss}$, the expression of the power gain for the $u$-th user is:
    \begin{equation}
        \dfrac{\mathbf{w}_u^T \mathbf{h}_u \mathbf{h}_u^H \mathbf{w}_u^*}{\mathbf{w}_u^T \mathbf{Z}_R \mathbf{w}_u^*}
    \end{equation}
   where \begin{equation}
\mathbf{Z}_R=\mathbf{Z}+\frac{r_{\rm loss}}{R_{\rm rad}}\mathbf{I}_M.
\end{equation}
\end{proposition}
\begin{proof}
\noindent\emph{Proof:} The radiation efficiency of an antenna array is defined as:
\begin{small}
\begin{equation}\label{effi}
\begin{aligned}
\eta&\triangleq\frac{P_{\rm rad}}{P_{\rm rad}+P_{\rm loss}}\\
&=\frac{\mathbf{w}_u^T R_{\rm rad}\mathbf{Z} \mathbf{w}_u^*}{\mathbf{w}_u^T R_{\rm rad}\mathbf{Z} \mathbf{w}_u^*+\mathbf{w}_u^Tr_{\rm loss}\mathbf{I}_M\mathbf{w}_u^*},
\end{aligned}
\end{equation}
\end{small}
where $P_\text{loss}$ is the array loss power, which is given by\cite{Collin_Zucker_1969}
\begin{equation}
   P_\text{loss} = \frac{1}{2}
\mathbf{w}_u^T r_{\operatorname{loss}} \mathbf{I}_M\mathbf{w}_u^*,
\end{equation}
where $r_{\rm loss}$ is the ohmic loss resistance of the antenna. Therefore, when considering the ohmic loss of the antenna, combining \eqref{Du} and \eqref{effi}, the gain $G_u$ of user $u$ is:
\begin{equation}
\begin{aligned}
G_u&=\eta D_u\\
&=\frac{\mathbf{w}_u^T \mathbf{h}_u \mathbf{h}_u^H\mathbf{w}_u}{\mathbf{w}_u^T (\mathbf{Z}+\frac{r_{\rm loss}}{R_{\rm rad}}\mathbf{I}_M)\mathbf{w}_u^*}.
\end{aligned}
\end{equation}
That is,
\begin{equation}
    \mathbf{Z}_R = \mathbf{Z}+
\frac{r_{\text {loss }}}{R_{\mathrm{rad}}} \mathbf{I}_M.
\end{equation}

which proves \textbf{Proposition} \ref{propositionohimc}.
\end{proof}

With the introduction of regularization, the optimization problem is defined as follows:
\begin{equation}\label{P2}
\begin{array}{ll}
\max\limits_{\mathbf{w}^r_u} & \dfrac{(\mathbf{w}_u^{\rm r})^T \mathbf{h}_u \mathbf{h}_u^H (\mathbf{w}_u^{\rm r})^*}{(\mathbf{w}_u^{\rm r})^T \mathbf{Z}_R (\mathbf{w}_u^{\rm r})^*} \\
\text { s.t. } & \mathbf{h}_i^T \mathbf{w}_u^{\rm r}=0, i=1, \cdots, U, i \neq u.
\end{array}
\end{equation}

For a dipole antenna, the ohmic loss resistance $r_{\rm loss}$ is\cite{balanis2012antenna}:
\begin{equation}\label{rloss}
r_{\mathrm{loss}}=\frac{L}{4 \pi r} \sqrt{\frac{\pi f \mu}{ \sigma}},
\end{equation}
where $L$, $r$, $f$, and $\sigma$ are the length, radius, frequency, and conductivity of the dipole antenna, and $\mu$ is the vacuum permeability. The radiation resistance is approximately:
\begin{equation}\label{Rrad}
R_{r a d} \approx 73 \text{ ohms}.
\end{equation}
 Note that \eqref{P2} has the same form as problem \eqref{P1}, the methods proposed in \textbf{Algorithm} \ref{alg2} can be directly applied to solve it.
\section{Wideband multi-user superdirective precoding}\label{IV}
The superdirective antenna array was previously thought to have a narrow directivity bandwidth since antenna designers generally employed the same beamforming vector (excitation) for the whole target bandwidth. However, in communications systems like orthogonal frequency division multiplexing (OFDM), wideband beamforming can be achieved by applying different beamforming vectors at different frequencies. The reason of the narrow directivity bandwidth caused by using the same beamforming vector at all frequencies can be explained by the following proposition. Firstly, define $f_l$ as the frequency at which the power gain experiences a 3dB decrease, and is lower than the center frequency $f_c$.
\begin{proposition}\label{propositionwideband}
Let $f_c$ denote the center frequency at which the superdirective linear antenna array operates. The array achieves its maximum directivity gain $D_c$ at $f_c$, using the optimal beamforming vector $\mathbf{a}_c$. 
The following relationship holds:
\begin{equation}
	\frac{f_c-f_l}{f_c} \propto \frac{1}{M^2},
\end{equation}
where $M$ is the number of antennas.
\end{proposition}
\begin{proof}
\noindent\emph{Proof:} For antennas with one narrow
major lobe and negligible minor lobes, the relationship between directivity $D$ and beamwidth  can be approximated as\cite{balanis2012antenna}
\begin{equation}\label{Dsim}
    D\simeq \frac{4\pi}{\theta_{\rm HP}\phi_{\rm HP}},
\end{equation}
where $\theta_{\rm HP}$ and $\phi_{\rm HP}$ are the 3dB beamwidths in the E-plane and H-plane, respectively. 

Generally speaking, the radiation pattern of a superdirective antenna array has only one main lobe, and the sidelobes are very small, which enables an approximate relationship between the directivity and the beamwidth using \eqref{Dsim}. Additionally, as the superdirective linear antenna array operates in the end-fire direction, the radiation pattern possesses rotational symmetry, such that $\theta_{\rm HP}=\phi_{\rm HP}$. Therefore, based on \eqref{Dsim}, the 3dB beamwidth in a superdirective antenna array can be approximated by
\begin{equation}
    \phi_{\rm HP}\simeq  \sqrt\frac{{4 \pi}}{D}.
\end{equation}
Denote $\lambda_c$ and $\lambda_l$ the wavelengths corresponding to $f_c$ and $f_l$, respectively. For the maximum radiation direction $\phi_p$, the steering vector $\mathbf{f}_c\left(\phi_{p}\right)$ at the center frequency is given by
\begin{equation}
    \mathbf{f}_c\left(\phi_{p}\right)=\frac{1}{\sqrt{M}}\left[1, e^{j \pi \phi_{c}}, e^{j \pi 2 \phi_{c}}, \cdots, e^{j \pi\left(M-1\right) \phi_{c}}\right]^T,
\end{equation}
where $\phi_c=\frac{2d}{\lambda_c}\cos(\phi_p)$. The steering vector corresponding to the half-power point is 
\begin{equation}
    \mathbf{f}_c\left(\phi_{p}+\frac{\phi_{\rm HP}}{2}\right)=\frac{1}{\sqrt{M}}\left[1, e^{j \pi \phi'_{c}}, e^{j \pi 2 \phi'_{c}}, \cdots, e^{j \pi\left(M-1\right) \phi'_{c}}\right]^T,
\end{equation}
where $\phi'_c=\frac{2d}{\lambda_c}\cos(\phi_p+\frac{\phi_{\rm HP}}{2})$. 
For the same maximum radiation direction $\phi_p$, as the frequency changes from $f_c$ to $f_l$, the steering vector becomes
\begin{equation}
   \mathbf{f}_l\left(\phi_{p}\right)=\frac{1}{\sqrt{M}}\left[1, e^{j \pi \phi_{l}}, e^{j \pi 2 \phi_{l}}, \cdots, e^{j \pi\left(M-1\right) \phi_{l}}\right]^T. 
\end{equation}
where $\phi_l=\frac{2d}{\lambda_l}\cos(\phi_p)$. The phenomenon of the steering vector changing with frequency is referred to as the beam split effect\cite{dai2022delay}. Thus, the half-power frequency  $f_l$ can be deduced from the following equation:
\begin{equation}\label{sineq}
    \frac{2d}{\lambda_c}\cos(\phi_p+\frac{\phi_{\rm HP}}{2})=\frac{2 d}{\lambda_l} \cos \left(\phi_p\right),
\end{equation}
and
\begin{equation}
\begin{aligned}\label{sinsim}
    \cos(\phi_p+\frac{\phi_{\rm HP}}{2})&=\cos\phi_p\cos{\frac{\phi_{\rm HP}}{2}}-\sin{\phi_p}\sin\frac{\phi_{\rm HP}}{2}\\
    &\simeq\cos\phi_p\cos{\frac{\phi_{\rm HP}}{2}}.
\end{aligned}
\end{equation}
This approximation is valid since the superdirective linear antenna array operates in the end-fire direction, i.e., near $0^{\circ}$ or $180^\circ$, therefore $|\cos(\phi_p)|\gg|\sin(\phi_p)|$. And for highly directional radiation patterns, the 3dB beamwidth is very narrow, hence $\sin(\frac{\phi_{\mathrm{HP}}}{2})\approx0$. Combining \eqref{sineq} and \eqref{sinsim}, we have
\begin{equation}
\begin{aligned}
    f_l&=f_c\cos\frac{\phi_{\rm HP}}{2}\\
    &=f_c\cos(\frac{1}{2}\sqrt\frac{{4 \pi}}{D}).
\end{aligned}
\end{equation}
Since the maximum directivity of the superdirective antenna array is $M^2$, thus
\begin{small}
\begin{equation}
\begin{aligned}
    f_c-f_l&=f_c-f_c\cos(\frac{\sqrt{4\pi}}{2M})\\
    &=f_c(1-\cos(\frac{\sqrt{4\pi}}{2M})\\
    &=f_c((\frac{\sqrt{4\pi}}{2M})^2\cdot \frac{1}{2!}+o(\frac{1}{16M^4}))\\
    &\simeq \frac{\pi f_c}{2M^2}.
\end{aligned}
\end{equation}
\end{small}
Therefore, we have
\begin{equation}
    \frac{f_c-f_l}{f_c} \propto \frac{1}{M^2},
\end{equation}
which proves \textbf{Proposition} \ref{propositionwideband}.
\end{proof}

\textit{Remarks:} \textbf{Proposition} \ref{propositionwideband} indicates that the relative distance between the center frequency and the half-power frequency is inversely proportional to the square of the number of antennas in superdirective antenna arrays. Specifically, the directivity bandwidth of a superdirective antenna array will significantly decrease with an increase in the number of antennas. To alleviate this problem, it is necessary to study wideband beamforming techniques.

Assuming that the bandwidth of the system is $B$ and its center frequency is $f_c$, and there are $N_f$ subcarriers. The set of frequencies of interest is $\Omega_f=\{f_1, f_2, \ldots, f_{N_f}\}$. For user $u$, the set of precoding vectors corresponding to each frequency  is given by:
\begin{equation}
\mathbf{W}_u=\mathrm{diag}\{\mathbf{w}_{f_1}^u,\mathbf{w}_{f_2}^u,\ldots,\mathbf{w}^u_{f_{N_f}}\},
\end{equation}
where $\mathbf{w}_{f_i}^u$ denotes the precoding vector for user $u$ at frequency  $f_i$.

The corresponding set of channels is given by:

\begin{equation}
\mathbf{H}_u=\mathrm{diag}\{\mathbf{h}_{f_1}^u,\mathbf{h}_{f_2}^u,\ldots,\mathbf{h}_{f_{N_f}}^u\},
\end{equation}
where $\mathbf{h}_{f_i}^u$ denotes the channel between the transmitter and user $u$ at frequency  $f_i$.
The impedance coupling matrix $\mathbf{Z}(f_i)$ is a matrix whose elements are given by:
\begin{equation}
z^{(f_i)}_{m n}=\left\{\begin{array}{cc}
\frac{\sin [k_i d(n-m)]}{k_i d(n-m)}, & n \neq m \\
1, & n=m
\end{array}\right.,
\end{equation}
where $k_i=\frac{2\pi}{\lambda_i}=\frac{2\pi f_i}{\mathrm{c}}$,  $\mathrm{c}$ is the speed of light, and $\lambda_i$ is the wavelength corresponding to frequency $f_i$.

Thus, under the wideband condition, the optimization problem is formulated as:
\begin{equation}
\begin{aligned}\label{maxW}
\max&\quad \sum_{i=1}^{N_f}\frac{(\mathbf{w}_{f_i}^u)^T(\mathbf{h}_{f_i}^u)^H\mathbf{h}_{f_i}^u(\mathbf{w}_{f_i}^u)^*}{(\mathbf{w}_{f_i}^u)^T\mathbf{Z}(f_i)(\mathbf{w}_{f_i}^u)^*}  \\
\text { s.t. } & \mathbf{H}_j^T \mathbf{\mathbf{W}_u}=\mathbf{0}, j=1, \cdots, U, j \neq u
\end{aligned}
\end{equation}
Due to the independence between frequencies, the optimal precoding vector for each frequency can be obtained directly using the method proposed in \textbf{Algorithm }\ref{alg2}, resulting in a collection of beamforming vectors that form the optimal solution to the problem \eqref{maxW}.

\section{Numerical results}\label{V}
In this section, we evaluate the performance of our proposed methods using numerical analysis software and full-wave electromagnetic simulation software. Two well-known multi-user precoding schemes, namely MRT  and ZF precoding, are used as benchmarks for the comparison of the spectral efficiency under multi-user scenarios.

To verify the effectiveness of the proposed field coupling aware channel estimation method, we first compute the field-coupling matrix of an antenna array using the method proposed in \cite{han2023superdirective}. We set the number of antennas and the spacing between them to $M=8$ and $d=0.2\lambda$, respectively. The number of users is set to $U=5$, and the arrival angles of each user are uniformly distributed in $[0, \pi]$. The channel $\mathbf{h}_i,i=1,\ldots,U,$ of each user is modeled using \eqref{multipath_h}. By replacing the steering vector $\mathbf{e}(\theta_{up},\phi_{up})$ in \eqref{multipath_h} with the distorted steering vector obtained from the full-wave simulation software, we obtain the distorted channel $\mathbf{h}_i^\mathrm c,i=1,\ldots,U,$ for each user affected by field coupling. Specifically, the true received signal for user $u$ can be expressed as:
\begin{equation}
   \mathbf{y}_u = \mathbf{h}_u^\mathrm c\mathbf{s}_u^T+\mathbf{N},
\end{equation}
where the elements of $\mathbf{N}$ are distributed as complex Gaussian random variables with mean $0$ and variance $\epsilon^2$. We estimate the channel $\mathbf{h}_u$ with the true received signal using the proposed FCA estimation method and the traditional MMSE method respectively, and evaluate the performance using the normalized channel estimation error defined as:
\begin{equation}
\operatorname{err} \triangleq 10 \log _{10}\left(\frac{1}{U}\sum_{i=1}^U\frac{\left\|\widehat{\mathbf{h}}_{i}-\mathbf{h}_{i}\right\|_F^2}{\left\|\mathbf{h}_{i}\right\|_F^2}\right)
\end{equation}
\begin{figure}[tbp]
  \centering
  \includegraphics[width=0.5\textwidth]{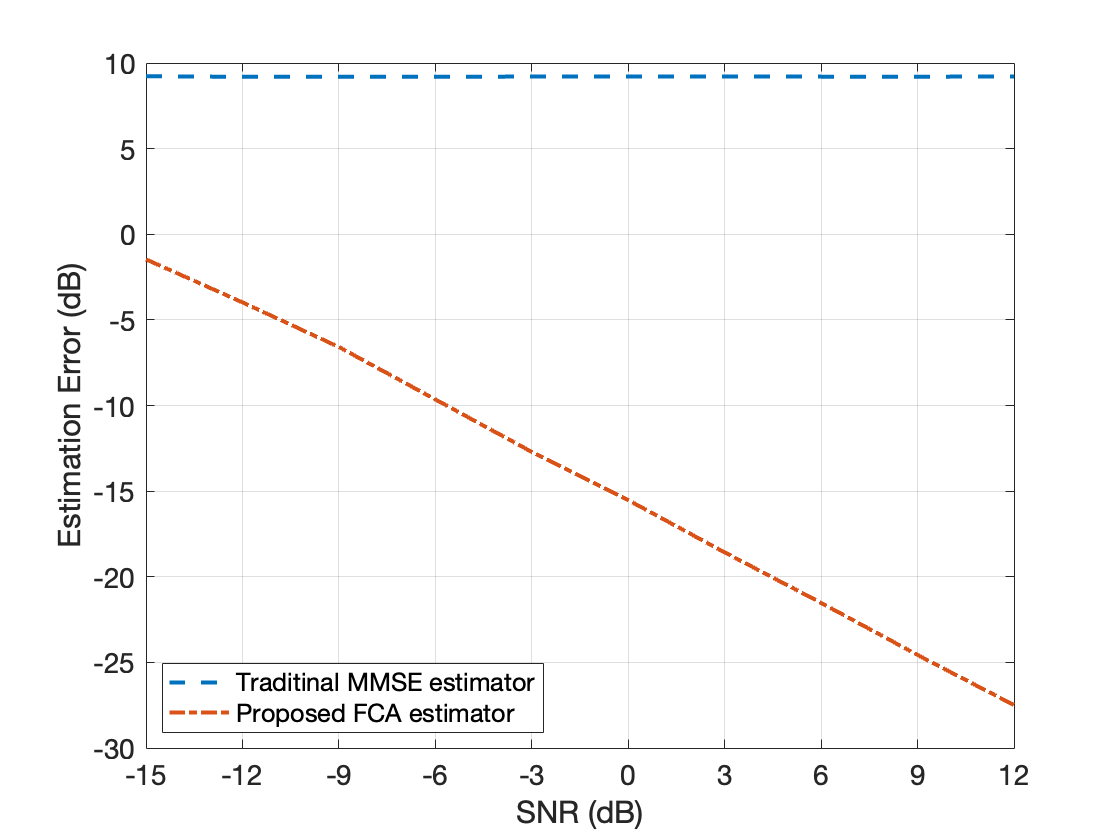}
  \caption{The channel estimation errors  vs. SNR.}
  \label{channel_est}
\end{figure}

We define $\mathrm{SNR}=\frac{\tau}{\epsilon^2}$, where $\tau$ is the pilot power, which is assumed common for all users. We simulate channel estimation errors at different SNRs, performing 2000 random channel realizations for each SNR level, and present the average results in Fig. \ref{channel_est}. We can see that the performance of the traditional estimator deteriorates due to its inability to account for the strong field coupling caused by narrow antenna spacing. The estimation error remains almost unchanged with increasing SNR, indicating that the traditional estimator cannot accurately estimate the channel in practical strong coupling arrays. On the other hand, our proposed FCA estimator is capable of handling channel estimation in compact arrays, with significantly reduced estimation errors compared to the traditional estimator. Moreover, the FCA estimator exhibits noticeable performance improvement with increasing SNR.

\begin{figure}[tbp]
\centering
\includegraphics[width=0.5\textwidth]{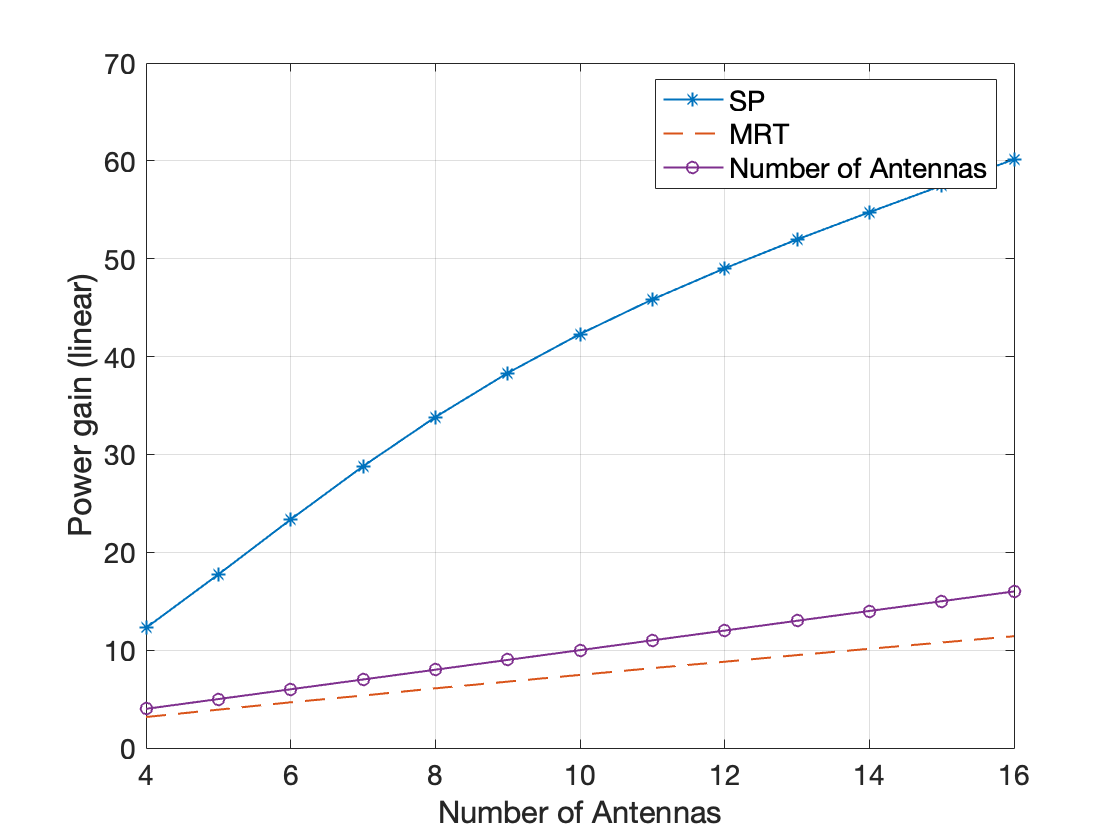}
  \caption{Comparison of power gain between superidirective precoding (SP) and MRT in a multipath scenario.}
  \label{multipath_comp}
\end{figure}

Fig. \ref{multipath_comp} shows the power gain vs. the number of antennas in a multipath scenario under different precoding strategies. The simulation is set to a single-user communication scenario with four paths and the channel model is given by \eqref{multipath_h}. The channel gains of all paths are normalized to 1. Since the superdirectivity antenna array works in the vicinity of end-fire direction, the arrival angle of each path is set to uniformly distributed in the range of $[0^\circ, 20^\circ]$ in the simulation. The results in Fig. \ref{multipath_comp} is the mean of 1000 random trials. As shown in Fig. \ref{multipath_comp}, the power gain under the SP strategy is significantly higher than that under the MRT strategy. However, since the multipath scenario contains paths that are not in the end-fire direction, the power gain does not reach $M^2$. In addition, it can be observed that the power gain of MRT is close to the number of antennas, which verifies the common sense that MRT can only achieve a power gain of up to $M$.

\begin{figure}[tbp]
  \centering
\includegraphics[width=0.5\textwidth]{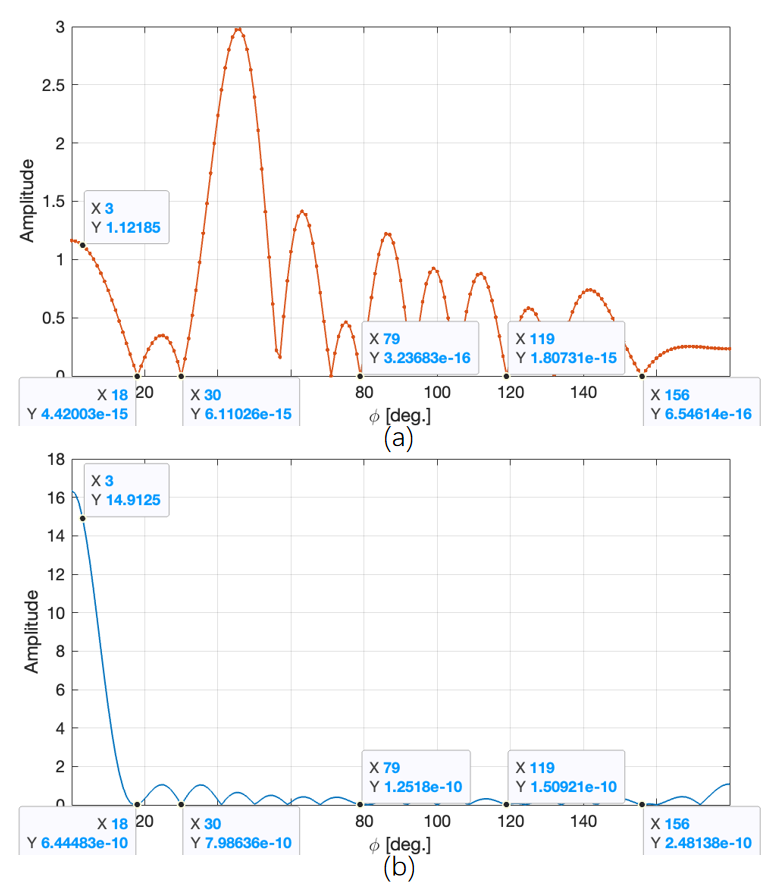}
  \caption{Radition pattern generated by (a) ZF precoding method and (b) proposed INSP precoding method.}
  \label{musers}
\end{figure}
Then, the multi-user simulations are conducted to investigate the performance of the proposed methods. Two classic multi-user precoding methods, namely MRT and ZF precoding, serve as benchmarks. Their expressions are respectively:
\begin{equation}
    \mathbf{W}^{\rm{MRT}} = \gamma_{\rm m}\mathbf{H}^H
\end{equation}
and
\begin{equation}
    \mathbf{W}^{\rm{ZF}} = \gamma_{\rm z}\mathbf{H}^H\left(\mathbf{H} \mathbf{H}^H\right)^{-1},
\end{equation}
where $\mathbf{H} = [\mathbf{h}_1,\mathbf{h}_2,\cdots,\mathbf{h}_U]$ is the channel matrix, the two scalars $\gamma_{\rm m}$ and $\gamma_{\rm z}$ are determined by the power constraints. In  simulations, the transmit power is equally allocated to each user. As the radiation resistance of the antenna is a fixed value, according to \eqref{Prad}, the transmit power is constrained using the following expression:
\begin{equation}
\mathbf{a}^H_u\mathbf{Z}\mathbf{a}_u=1,u=1,2,\cdots,U,
\end{equation}
where $\mathbf{a}_u$ is the precoding vector for the $u$-th user.
The performance metric is the spectral efficiency defined as
\begin{equation}
\mathrm{SE}\triangleq\sum_{u=1}^U \log _2\left(1+\frac{\left|\mathbf{h}_u^T \mathbf{a}_u\right|^2}{\sum_{j \neq u}\left|\mathbf{h}_j^T \mathbf{a}_j\right|^2+\epsilon^2}\right),
\end{equation}

In the multi-user simulation scenario, we set the number of antennas as $M=20$ with an antenna spacing of $d=0.25\lambda$.  To validate Algorithm \ref{alg2}, without loss of generality, we set the number of users as $U=6$, with the arrival angles of $\phi=3^\circ,18^\circ,30^\circ,79^\circ,156^\circ$, and $119^\circ$, respectively for each user. The target user has an arrival angle of $\phi=3^\circ$.  The radiation patterns obtained using both ZF precoding method and the proposed INSP precoding method are shown in Fig. \ref{musers}. It can be observed that both precoding methods achieve nulling of interfering users. However, the radiation pattern obtained by the ZF precoding method has larger sidelobes, and the received signal power at the target user is relatively small. On the other hand, our proposed INSP method utilizes the superdirectivity of the antenna array under strong coupling, significantly improving the received signal power at the target user.

Fig. \ref{muser20_8} depicts the  spectral efficiency as a function of SNR for different precoding schemes, where the number of users is set to 8 and the SNR is defined as $\mathrm{SNR}\triangleq\frac{U}{\epsilon^2}$. Considering that the superdirective antenna array works in the end-fire direction, the arrival angles of the users are set to be in the vicinity of the end-fire direction, i.e., around $0^\circ$ and $180^\circ$. 
In the simulation, the arrival angle distribution of four users satisfies a uniform distribution within $[(10(i-1))^\circ,(10i)^\circ], i=1,\cdots,4$, and the arrival angle distribution of the other four users satisfies a uniform distribution within $[(140+10(i-1))^\circ,(140+10i)^\circ], i=1,\cdots,4$. The results in Fig. \ref{muser20_8} is the mean of 1000 random channel realizations.

\begin{figure}[tbp]
  \centering
\includegraphics[width=0.5\textwidth]{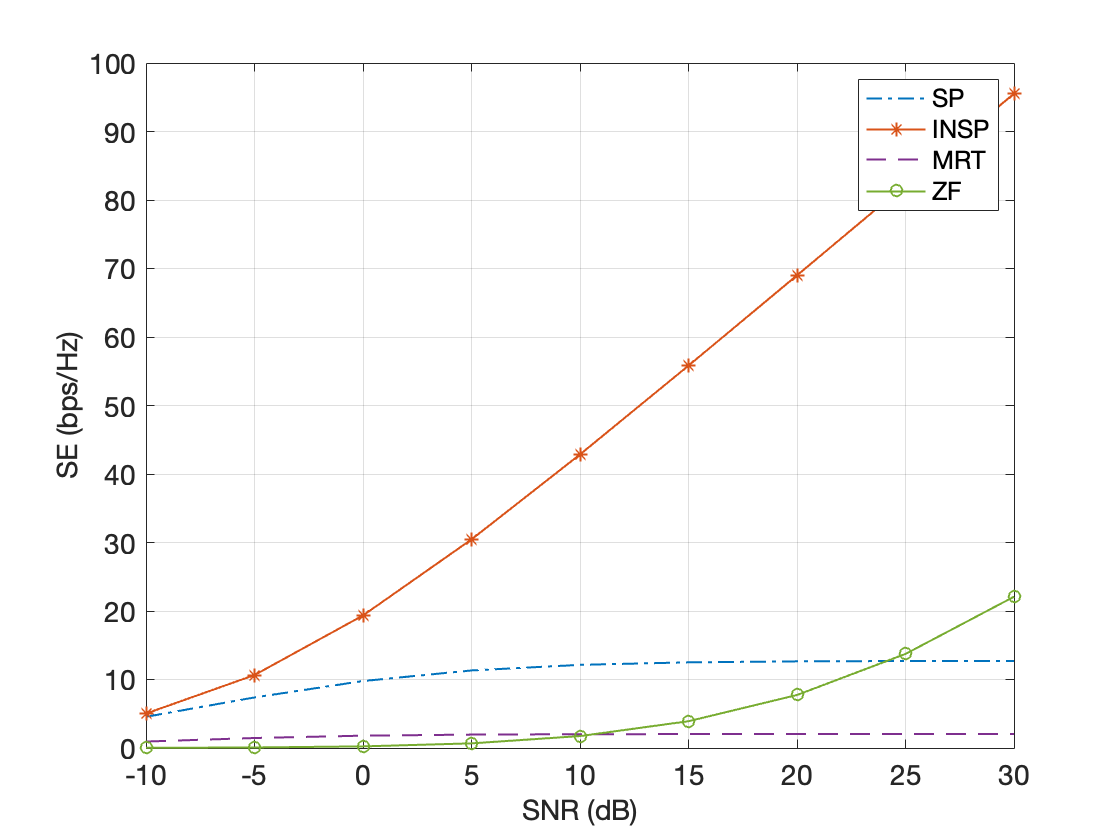}
  \caption{Spectral efficiency vs. SNR for various precoding schemes with 20 antennas, 0.25$\lambda$ antenna spacing, and 8 users.}
  \label{muser20_8}
\end{figure}

In Fig. \ref{muser20_8}, it can be seen that the proposed SP scheme performs well at low SNR, while the interference from the users grows with the increase of SNR.   
The INSP scheme performs the best among all schemes, due to the interference-nulling ability. It can be seen that at high SNR, the spectral efficiency achievable by the INSP scheme is 4 to 5 times that of the ZF scheme, demonstrating the capability of the superdirective multi-user schemes to significantly improve communication rate.

\begin{figure}[tbp]
  \centering
\includegraphics[width=0.5\textwidth]{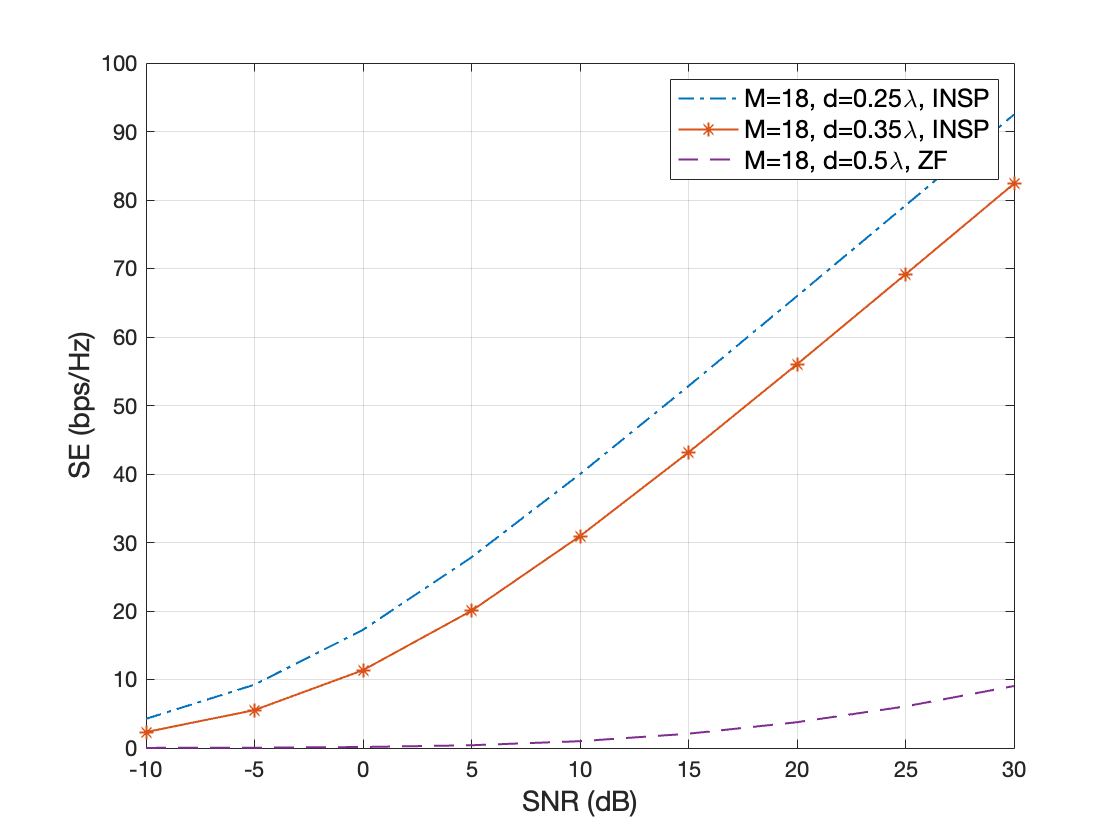}
  \caption{Spectral efficiency vs. SNR for different antenna array apertures with the same number of antennas, where the number of users is 8.}
  \label{number}
\end{figure}

In Fig. \ref{number}, the spectral efficiency is evaluated for different array apertures while maintaining the same number of antennas. From the figure, it is evident that by employing a superdirective antenna array with a reduced aperture ($M=18, \, d=0.25\lambda$), a remarkable improvement in spectral efficiency, up to 9 times greater than that of a traditional MIMO array ($M=18, \, d=0.5\lambda$), can be achieved. As a result, superdirectivity antenna arrays  can improve spectral efficiency, and in the meantime  greatly reduce the array aperture.

Next, we consider  a superdirective antenna array with ohmic loss. The antenna element used is the dipole antenna operating at $1.6$  GHz, with a copper material having a conductivity of $\sigma=5.8\times 10^7$ S/m and a magnetic permeability of $4\pi\times 10^{-7}$ H/m. The length and radius of the antenna are $85$ mm and $0.75$ mm, respectively. According to \eqref{rloss} and \eqref{Rrad}, the radiation resistance and loss resistance of the antenna element can be calculated, respectively. The antenna array consists of 20 elements with a spacing of 0.25$\lambda$. The number of users is 8, and the arrival angle distribution is the same as the simulation conditions in Fig. \ref{muser20_8}.  Considering the  ohmic loss, we enforce a constraint on the transmit power using $\mathbf{a}^H_u\mathbf{Z}_R\mathbf{a}_u=1,u=1,\cdots,U$. The simulation results are shown in Fig \ref{lossmuser20_8}. It can be seen that by incorporating the impact of ohmic loss into the model for optimization, the performance of the RINSP scheme is the best among all schemes, achieving a spectral efficiency close to twice that of the ZF scheme at high SNR. 
However, the gain is not as significant as in Fig. \ref{number}. This also reflects that ohmic loss is an important factor that limits the implementation of superdirective antenna arrays in practical applications. In the future, this problem may be addressed by superconducting antennas.

\begin{figure}[tbp]
  \centering
  \includegraphics[width=0.5\textwidth]{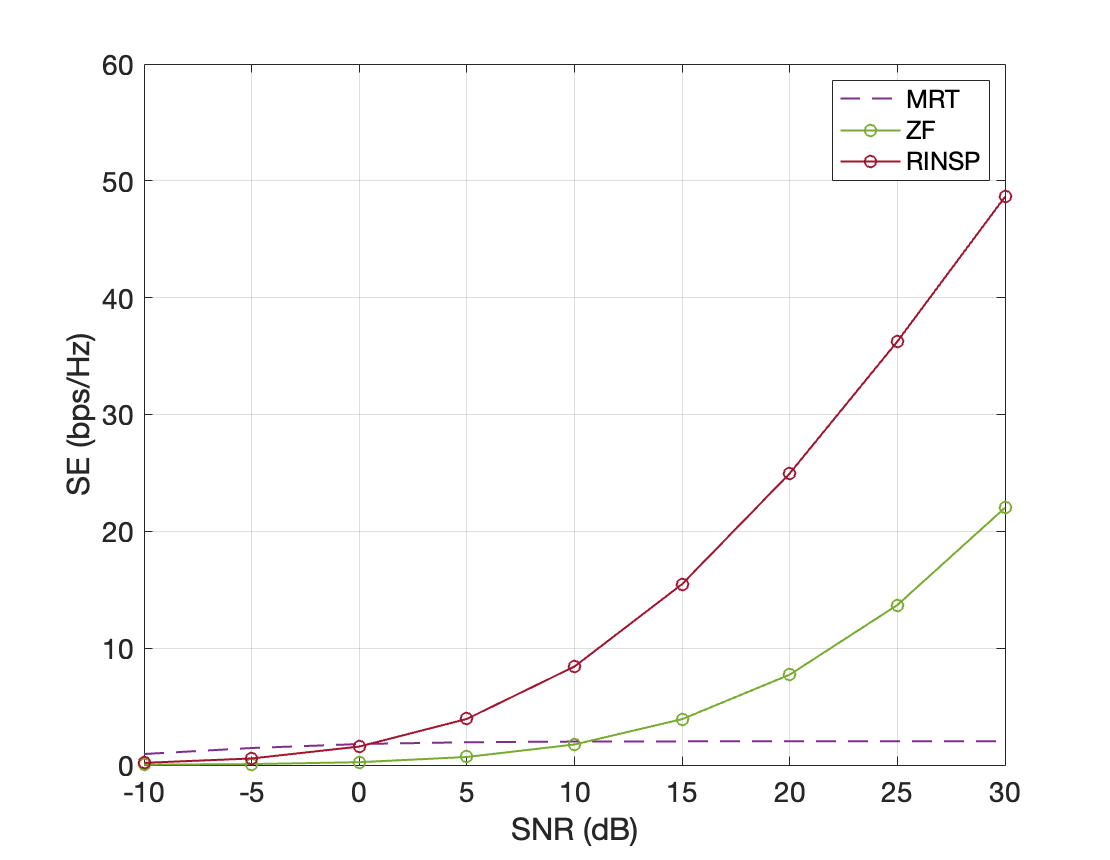}
  \caption{Considering the ohmic loss, spectral efficiency vs. SNR for various precoding schemes with 20 antennas and 0.25$\lambda$ antenna spacing, 8 users.}
  \label{lossmuser20_8}
\end{figure}

\begin{figure}[tbp]
  \centering
\includegraphics[width=0.5\textwidth]{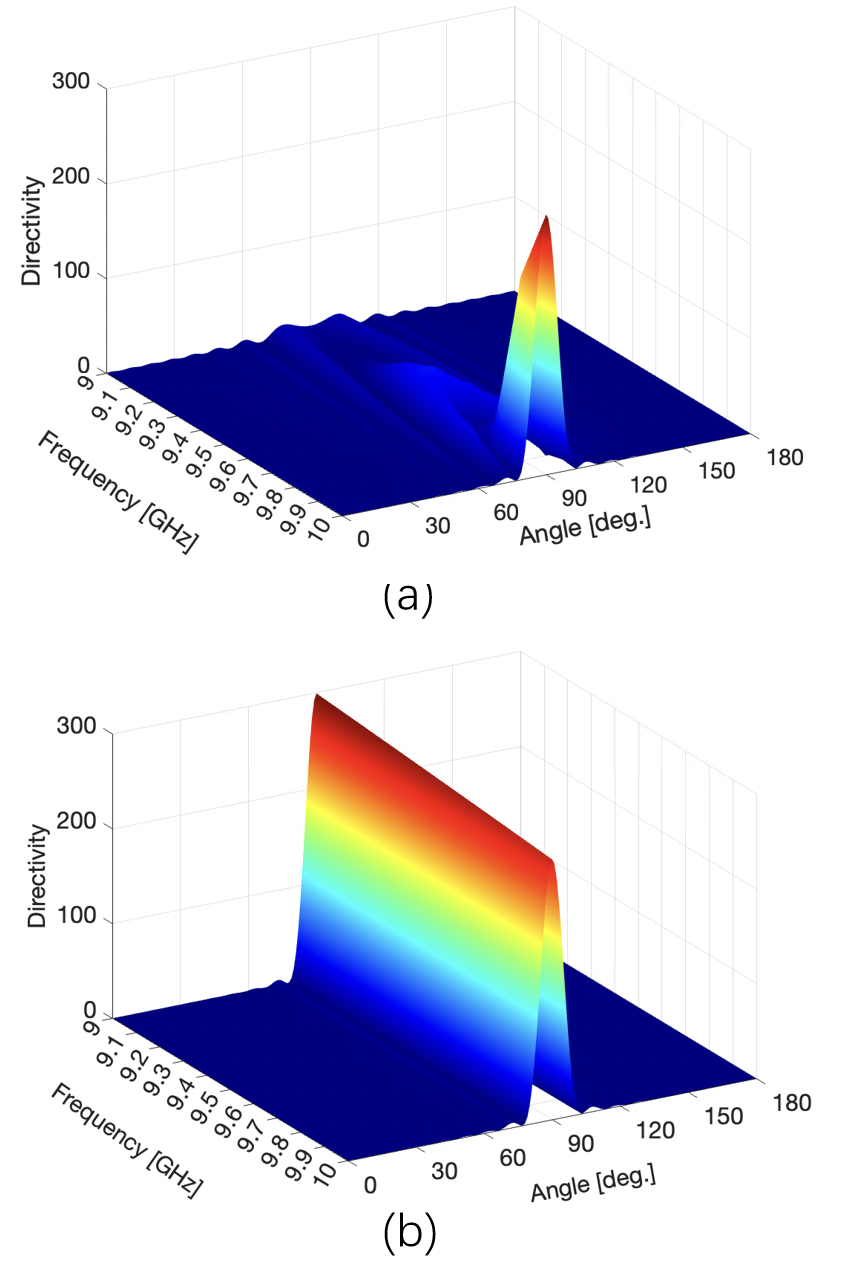}
  \caption{Directivity patterns (a) without wideband beamforming and (b) with wideband beamforming.}
  \label{broadband}
\end{figure}

In Fig. \ref{broadband}, we demonstrate the narrow-band characteristics of the superdirective antenna array and the importance of using wideband beamforming. Fig. \ref{broadband} (a) illustrates the directivity pattern produced by applying the optimal beamforming coefficients at 10 GHz to the frequency range of 9-10 GHz. As expected, the maximum directivity gain sharply decreases as the frequency deviates from 10 GHz, consistent with \textbf{Proposition} \ref{propositionwideband}. Using the superdirective beamforming  proposed in Section \ref{IV}, the directivity pattern obtained is shown in Fig. \ref{broadband} (b). It can be observed that a maximum directivity gain of approximately 300 is achieved at each frequency, validating the effectiveness of the proposed method.

\begin{figure}[tbp]
  \centering
\includegraphics[width=0.5\textwidth]{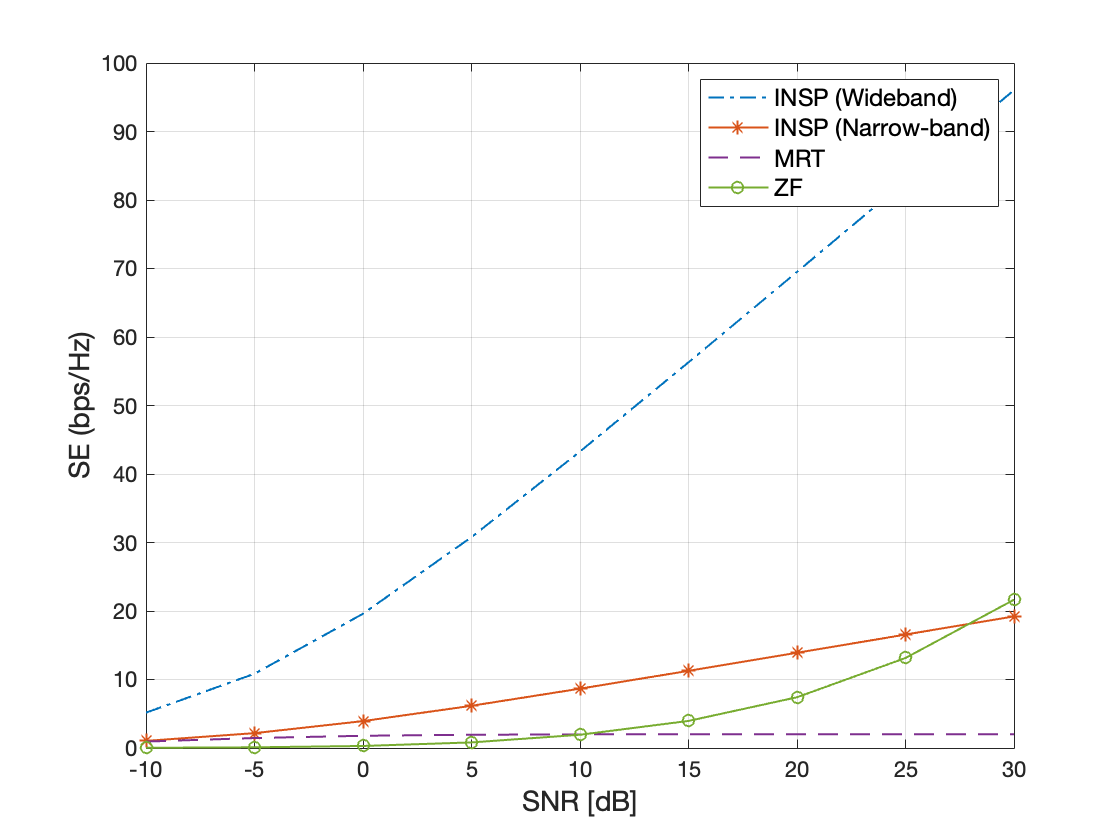}
  \caption{Comparison of the spectral efficiency using the wideband precoding and the narrow-band precoding.}
  \label{widebandcomp}
\end{figure}
In Fig. \ref{widebandcomp}, the spectral efficiencies of the wideband superdirective method and the narrow-band superdirective beamforming method are shown, where the number of subcarriers  is $N_f=5$, the center frequencies are $[9.4,9.7,10,10.3,10.6]$ GHz, the number of antennas is 20, with a spacing of 0.25$\lambda$, and the number of users is 8. The wideband superdirective precoding vector is calculated using the method proposed in section \ref{IV}, while the narrow-band precoding applies the precoding vector at 10 GHz to all frequencies. It can be seen that, compared to wideband beamforming approach, using narrowband beamforming in superdirective arrays results in a significant reduction in spectral efficiency. This phenomenon can also be interpreted by Fig. \ref{broadband}. The use of wideband superdirective precoding can enhance directivity bandwidth, thereby improving the overall spectral efficiency.

\section{Conclusion}
In this paper, we studied the superdirectivity-enhanced multi-user wireless communications. By taking into account the mutual coupling effect in compact arrays, the proposed field-coupling-aware channel estimation method can accurately estimate the coupling channel. For scenarios with multiple users, we propose a series of superdirective multi-user precoding methods targeting various use cases.
The SP method extends our previous single-user single-path coupling-based superdirective beamforming approach to the case of multipath for multiple users.
The INSP method represents the optimal solution to maximize power gain for each user without creating interference to other users. 
The RINSP method takes into account the ohmic losses in the array. Finally, a wideband precoding approach method was introduced, addressing the narrow bandwidth issue of superdirective arrays. Simulation results demonstrated that the proposed methods significantly outperform the state-of-the-art methods. In particular, our study showed the superdirective antenna array assisted communication system may achieve up to 9 times higher spectral efficiency compared with traditional MIMO, while reducing the array aperture by half. 

\appendices
\section{Proof of Lemma 1}\label{Appen. A}
The form of \eqref{D2} is that of a generalized Rayleigh quotient, and its optimization problem can be solved by means of the following generalized eigenvalue problem\cite{li2015rayleigh}:
\begin{equation}\label{gep}
\mathbf{Z}^{-1} \mathbf{h}_u \mathbf{h}_u^H (\mathbf{w}^{\rm SP}_u)^*=D(\mathbf{w}^{\rm SP}_u)^*.
\end{equation}
The maximum eigenvalue of the above eigenvalue problem corresponds to the maximum value of \eqref{Du}, and the corresponding eigenvector is the optimal precoding vector. Based on the property of rank relationships in matrix multiplication, we have
\begin{align}
    \rm{rank}(\mathbf{Z}^{-1}\mathbf{h}_u\mathbf{h}_u^H)\leq \min\{\rm{rank}(\mathbf{Z}^{-1}),\rm{rank}(\mathbf{h}_u\mathbf{h}^H_u)\}.
\end{align}
Since $\rm{rank}(\mathbf{h}_u\mathbf{h}^H_u)=1$, the relationship holds:
\begin{align}
    \rm{rank}(\mathbf{Z}^{-1}\mathbf{h}_u\mathbf{h}_u^H)\leq 1.
\end{align}
In addition, if $\rm{rank}(\mathbf{Z}^{-1}\mathbf{h}_u\mathbf{h}_u^H)=0$, it implies that the maximum value of \eqref{Du} is 0, which is not consistent with reality, therefore
\begin{align}
    \rm{rank}(\mathbf{Z}^{-1}\mathbf{h}_u\mathbf{h}_u^H)= 1.
\end{align}
As a result,  \eqref{gep} has only one non-zero eigenvalue, which is equal to the maximum value of  \eqref{Du}.
Since
\begin{equation}
    \begin{aligned}\label{zhhw}
    \mathbf{Z}^{-1} \mathbf{h}_u \mathbf{h}_u^H (\mathbf{w}^{\rm SP}_u)^*&=\mathbf{Z}^{-1} \mathbf{h}_u (\mathbf{h}_u^H (\mathbf{w}^{\rm SP}_u)^*)\\
    &=\gamma \mathbf{Z}^{-1} \mathbf{h}_u
\end{aligned}
\end{equation}

The second equality follows from the associative property of matrix multiplication. Combining \eqref{gep} and \eqref{zhhw}, we have
\begin{equation}
    \begin{aligned}
    \mathbf{w}^{\rm SP}_u&=\frac{\gamma}{D}\mathbf{Z}^{-1}\mathbf{h}_u^*\\
    &=\gamma_u\mathbf{Z}^{-1}\mathbf{h}_u^*,
\end{aligned}
\end{equation}
which proves \textbf{Lemma} \ref{lemma1}.\qed
\section{Proof of Theorem 1}\label{Appen. B}
\begin{proof}
Firstly, we perform Schmidt orthogonalization on the channels of all users except the target user $u$. Let the matrix $\left[ \boldsymbol{\nu }_1,\cdots ,\boldsymbol{\nu }_{U-1} \right] \in \mathbb{C} ^{M\times \left( U-1 \right)}$ consist of the channel vectors of the $U-1$ interfering users.  Thus, the constraint of problem \eqref{P1} can be rewritten as
\begin{equation}\label{vw0}
\boldsymbol{\nu }_i^T \mathbf{w}_u^{\mathrm{I}}=0, i=1, \cdots, U-1.
\end{equation}
The Schmidt orthogonalization process for channels of interfering users can be expressed as follows\cite{horn2012matrix}:
\begin{align}
    \left\{ \begin{array}{l}
	\bar{\boldsymbol{\nu}}_1=\frac{\boldsymbol{\nu }_1}{\parallel \boldsymbol{\nu }_1\parallel}\\
	\bar{\boldsymbol{\nu}}_2=\boldsymbol{\nu }_2-\frac{\bar{\boldsymbol{\nu}}_{1}^{H}\boldsymbol{\nu }_2}{\parallel \bar{\boldsymbol{\nu}}_1\parallel}\bar{\boldsymbol{\nu}}_1\\
	\bar{\boldsymbol{\nu}}_3=\boldsymbol{\nu }_3-\frac{\bar{\boldsymbol{\nu}}_{1}^{H}\boldsymbol{\nu }_3}{\parallel \bar{\boldsymbol{\nu}}_1\parallel}\bar{\boldsymbol{\nu}}_1-\frac{\bar{\boldsymbol{\nu}}_{2}^{H}\boldsymbol{\nu }_3}{\parallel \bar{\boldsymbol{\nu}}_2\parallel}\bar{\boldsymbol{\nu}}_2\\
	\vdots\\
	\bar{\boldsymbol{\nu}}_{U-1}=\boldsymbol{\nu }_{U-1}-\frac{\bar{\boldsymbol{\nu}}_{1}^{H}\boldsymbol{\nu }_{U-1}}{\parallel \bar{\boldsymbol{\nu}}_1\parallel}\bar{\boldsymbol{\nu}}_1-...-\frac{\bar{\boldsymbol{\nu}}_{U-2}^{H}\boldsymbol{\nu }_{U-1}}{\parallel \bar{\boldsymbol{\nu}}_{U-2}\parallel}\bar{\boldsymbol{\nu}}_{U-2}\\
\end{array} \right. 
\end{align}
    
The dimension of the interference user space is represented by $N(N\leq U-1)$, i.e., $N=\rm{rank}\{\left[ \boldsymbol{\nu }_1,\cdots ,\boldsymbol{\nu }_{U-1} \right]\}$. The above orthogonalization process will produce $N$ non-zero vectors, denoted by $\mathbf{v}_i, i=1,\cdots,N$. We next construct a unitary matrix $\mathbf{V}$ of size $M\times M$ such that the first $N$ columns of $\mathbf{V}$ form a standard orthogonal basis for the interference user channels. Thus, the first $N$ columns of $\mathbf{V}$ can be corresponded to the normalization of these non-zero vectors, i.e., $\dfrac{\mathbf{v}_i}{\|\mathbf{v}_i\|},i=1,\cdots,N$. And $\mathbf{V}$ satisfies
\begin{align}\label{vv}
\mathbf{V}\mathbf{V}^H=\mathbf{I}_M.
\end{align}
By means of $\mathbf{V}$, we can perform a basis transformation on problem \eqref{P1}. Specifically, we let
\begin{equation}\label{vw}
\mathbf{V}^T \mathbf{w}_u^{\mathrm{I}}=\left(\begin{array}{c}
\boldsymbol{\beta} \\
\boldsymbol{\alpha}_u
\end{array}\right),
\end{equation}
where $\boldsymbol{\beta}\in\mathbb{C}^{N\times 1}$ and $\boldsymbol{\alpha}_u\in\mathbb{C}^{(M-N)\times 1}$. 
According to \eqref{vw0}, we can derive the following equation for the orthogonalized channels of the interfering users:
\begin{small}
\begin{equation}
\left\{\begin{aligned}
\bar{\boldsymbol{\nu}}_1^T \mathbf{w}_u^{\mathrm{I}} & =\frac{\boldsymbol{\nu}_1^T \mathbf{w}_u^{\mathrm{I}}}{\left\|\boldsymbol{\nu}_1\right\|} \\
& =0, \\
\bar{\boldsymbol{\nu}}_2^T \mathbf{w}_u^{\mathrm{I}} & =\boldsymbol{\nu}_2^T \mathbf{w}_u^{\mathrm{I}}-\frac{\boldsymbol{\nu}_2^T \bar{\boldsymbol{\nu}}_1}{\left\|\bar{\boldsymbol{\nu}}_1\right\|} \bar{\boldsymbol{\nu}}_1^T \mathbf{w}_u^{\mathrm{I}} \\
& =0, \\
\vdots\\
\bar{\boldsymbol{\nu}}_{U-1}^T \mathbf{w}_u^{\mathrm{I}} & =\boldsymbol{\nu}_{U-1} \mathbf{w}_u^{\mathrm{I}}-\ldots-\frac{\boldsymbol{\nu}_{U-1}^T \bar{\boldsymbol{\nu}}_{U-2}}{\left\|\bar{\boldsymbol{\nu}}_{U-2}\right\|} \bar{\boldsymbol{\nu}}_{U-2}^T \mathbf{w}_u^{\mathrm{I}} \\
& =0.
\end{aligned}\right.
\end{equation}
\end{small}
Since the first $N$ columns of $\mathbf{V}$ are given by the non-zero vectors in $\overline{\boldsymbol{\nu}}_i, i=1, \cdots, U-1$, Eq. \eqref{vw} can be rewritten as
\begin{equation}\label{vww0}
\mathbf{V}^T \mathbf{w}_u^{\mathrm{I}}=\left(\begin{array}{c}
\mathbf{0}_{N}\\
\boldsymbol{\alpha}_u
\end{array}\right)
\end{equation}
Similarly, 
\begin{equation}\label{vh}
\mathbf{V}^H \mathbf{h}_u=\left(\begin{array}{c}
\boldsymbol{\gamma} \\
\boldsymbol{\eta}_u
\end{array}\right),
\end{equation}
where $\boldsymbol{\eta}_u \in \mathbb{C}^{(M-N) \times 1}, \boldsymbol{\gamma} \in \mathbb{C}^{N \times 1}$.
\begin{equation}\label{vzv}
\mathbf{V}^H \mathbf{Z} \mathbf{V}=\left(\begin{array}{ll}
\boldsymbol{\Upsilon} & \mathbf{\Psi} \\
\boldsymbol{\Psi} & \boldsymbol{\Xi}_u
\end{array}\right),
\end{equation}
where $\boldsymbol{\Xi}_u \in \mathbb{C}^{(M-N) \times(M-N)}$, and the matrices $\boldsymbol{\Psi}$ and $\boldsymbol{\Upsilon}$ are also block matrices.

Hence, the objective function in \eqref{P1} can be transformed to
\begin{small}
\begin{align}\label{transform}
    \dfrac{(\mathbf{w}^{\rm I}_u)^T \mathbf{h}_u \mathbf{h}_u^H (\mathbf{w}^{\rm I}_u)^*}{(\mathbf{w}^{\rm I}_u)^T \mathbf{Z} (\mathbf{w}^{\rm I}_u)^*} 
& =\frac{(\mathbf{w}^{\rm I}_u)^T \mathbf{VV}^H \mathbf{h}_u \mathbf{h}_u^H \mathbf{V} \mathbf{V}^H (\mathbf{w}^{\rm I}_u)^*}{(\mathbf{w}^{\rm I}_u)^T \mathbf{V} \mathbf{V}^H \mathbf{Z} \mathbf{V} \mathbf{V}^H (\mathbf{w}^{\rm I}_u)^*}\notag \\
& =\frac{\left(\mathbf{0},\boldsymbol{\alpha}_u^T \right)\left(\begin{array}{c}
\boldsymbol{\gamma}\\
\boldsymbol{\eta}_u 
\end{array}\right)\left(\boldsymbol{\gamma}^H,\boldsymbol{\eta}_u^H\right)\left(\begin{array}{c}
\mathbf{0}\\
\boldsymbol{\alpha}_u^* 
\end{array}\right)}{\left(\mathbf{0},\boldsymbol{\alpha}_u^T\right)\left(\begin{array}{ll}
\boldsymbol{\Upsilon} & \mathbf{\Psi} \\
\boldsymbol{\Psi} & \boldsymbol{\Xi}_u
\end{array}\right)\left(\begin{array}{c}
\mathbf{0}\\
\boldsymbol{\alpha}_u^* 
\end{array}\right)} \notag\\
& =\frac{\boldsymbol{\alpha}_u^T \boldsymbol{\eta}_u \boldsymbol{\eta}_u^H \boldsymbol{\alpha}_u^*}{\boldsymbol{\alpha}_u^T \mathbf{\Xi_u} \boldsymbol{\alpha}_u^*}.
\end{align}
\end{small}
The first equality holds due to \eqref{vv}, while the second equality holds due to \eqref{vw0}, \eqref{vh}, and \eqref{vzv}. Note that the constraint of problem \eqref{P1} has been contained in \eqref{transform}. Thus, problem \eqref{P1} can be transformed into an unconstrained convex problem, which is formulated as follows:
\begin{equation}
\max\limits_{\boldsymbol{\alpha}_u} \quad \frac{\boldsymbol{\alpha}_u^T \boldsymbol{\eta}_u \boldsymbol{\eta}_u^H \boldsymbol{\alpha}_u^*}{\boldsymbol{\alpha}_u^T \boldsymbol{\Xi} \boldsymbol{\alpha}_u^*}
\end{equation}
Based on the proof process of \textbf{Lemma} \ref{lemma1}, neglecting the scaling factor, the closed-form solution to this problem is given by
\begin{equation}\label{alphau}
\boldsymbol{\alpha}_u=\boldsymbol{\Xi}_u^{-1} \boldsymbol{\eta}_u^*.
\end{equation}
According to \eqref{vww0}, the optimal solution $\mathbf{w}^{\rm I}_u$ of the problem \eqref{P1} is obtained by
\begin{align}
    {\mathbf{w}}_u^{\mathrm{I}}=\mathbf{V}^*\left(\begin{array}{c}
\mathbf{0}_{N} \\
\boldsymbol{\alpha}_u
\end{array}\right)
\end{align}
From the expression of \eqref{Rint}, if matrix $\mathbf{R}_{\rm{int}}^u$ is decomposed into its eigenvalues sorted in descending order, then in \eqref{decomp}, the first $N$ columns of matrix $\mathbf{L}_u$ are the orthonormal basis vectors of $\mathbf{h}_i, i=1,\ldots,U, i\neq u$, and these columns satisfy the condition of unitary matrix $\mathbf{V}$. This proves \textbf{Theorem} \ref{theorem1}.
\end{proof}

\ifCLASSOPTIONcaptionsoff
  \newpage
\fi


\bibliographystyle{IEEEtran}

    \bibliography{mybib}

%








\end{document}